\documentclass[sigconf]{acmart}

\usepackage{booktabs} 
\usepackage{bbm}
\usepackage[utf8]{inputenc}
\usepackage{mathtools}

\usepackage{mmacells}

\usepackage{floatflt}
\usepackage{mdframed,framed}

\mmaDefineMathReplacement[≤]{<=}{\leq}
\mmaDefineMathReplacement[≥]{>=}{\geq}
\mmaDefineMathReplacement[≠]{!=}{\neq}
\mmaDefineMathReplacement[→]{->}{\to}[2]
\mmaDefineMathReplacement[⧴]{:>}{:\hspace{-.2em}\to}[2]
\mmaDefineMathReplacement{∉}{\notin}
\mmaDefineMathReplacement{∞}{\infty}

\newcommand{\ass}{\mathrel{:=}}     
\newcommand{\reserved}[1]{\textbf{#1}} 
\newcommand{\DO}{\reserved{do}}
\newcommand{\OD}{\reserved{end~while}}
\newcommand{\WHILE}{\reserved{while}}

\mmaSet{
  morefv={gobble=2},
  linklocaluri=mma/symbol/definition:#1,
  morecellgraphics={yoffset=1.9ex}
}

\setcopyright{rightsretained}

\acmConference[ISSAC2017]{The 42nd International Symposium on Symbolic and
  Algebraic Computation}{July 2017}{Kaiserslautern, Rheinland-Pfalz, Germany} 
\acmYear{2017}
\copyrightyear{2017}

\let\set\mathbbm

\newtheorem*{remark}{Remark}

\begin{document}
\title[Invariant Generation with Hypergeometric~Sequences]{Automated
  Generation of Non-Linear Loop Invariants Utilizing Hypergeometric~Sequences}

\author{Andreas Humenberger, Maximilian Jaroschek, Laura Kov{\'{a}}cs}
\authornote{All authors are supported by the ERC Starting Grant 2014
  SYMCAR 639270. We also acknowledge funding from the Wallenberg Academy Fellowship 2014 TheProSE,
the Swedish VR grant GenPro D0497701, and the Austrian FWF
research project RiSE S11409-N23.}
\affiliation{%
  \institution{Technische Universität Wien}
  \department{Institut für Informationssysteme 184}
  \streetaddress{Favoritenstraße 9--11}
  \city{Vienna} 
  \postcode{A--1040}
  \country{Austria}
}
\email{ahumenbe@forsyte.at}
\email{maximilian@mjaroschek.com}
\email{lkovacs@forsyte.at}

\renewcommand{\shortauthors}{A. Humenberger, M. Jaroschek, L. Kov{\'{a}}cs}

\begin{abstract}
  Analyzing and reasoning about safety properties of software systems becomes an
  especially challenging task for programs with complex flow and, in particular,
  with loops or recursion. For such programs one needs additional information,
  for example in the form of loop invariants, expressing properties to hold at
  intermediate program points.  In this paper we study program loops with
  non-trivial arithmetic, implementing addition and multiplication among numeric
  program variables. We present a new approach for automatically generating all
  polynomial invariants of a class of such programs. Our approach turns programs
  into linear ordinary recurrence equations and computes closed form solutions
  of these equations. The computed closed forms express the most precise
  inductive property, and hence invariant. We apply Gr\"obner basis computation
  to compute a basis of the polynomial invariant ideal, yielding thus a finite
  representation of all polynomial invariants.  Our work significantly extends
  the class of so-called P-solvable loops by handling multiplication with the
  loop counter variable.  We implemented our method in the Mathematica package
  Aligator and showcase the practical use of our approach.
\end{abstract}

%
%



\maketitle

\section{Introduction}
\subsection{Overview}
\label{sec:intro}
Analysis and verification of software systems requires non-trivial
automation. Automatic generation of program properties describing safety and/or
liveness is a key step to such automation, in particular in the presence of
program loops (or recursion). For programs with loops one needs additional
information, in the form of loop invariants or conditions on ranking functions.

In this paper we focus on loop invariant generation for programs with
assignments implementing numeric computations over scalar variables. Our
programming model extends the class of so-called P-solvable loops. Our work is
based on and extends results of~\cite{kapur,laura}, in particular it relies on
the fact that the set of polynomial invariants of P-solvable loops form a
polynomial ideal and we employ reasoning about C-finite and hypergeometric
sequences to determine algebraic dependencies. We show how to compute the ideal
of polynomial invariants of extended P-solvable loops as follows: we model
programs as a system of recurrence equations and compute closed form sequence
solutions of these recurrences. If these sequences are of a certain type, which
includes, among others, polynomials, rational functions, exponential and
factorial sequences, then we compute a set of generators of the polynomial
invariant ideal via Gröbner bases. We implemented our approach in the
Mathematica package Aligator~\cite{aligator} that is able to compute polynomial
loop invariants for programs that, to the best of our knowledge, no other
approach is able to handle.

This paper is organized as follows. In Section~\ref{sec:preliminaries}, we state
basic definitions and facts about the algebra of linear ordinary recurrence
operators as well as C-finite and hypergeometric sequences. We also give a
precise definition of the programming model we take into consideration,
particularly the notion of imperative loops with assignment statements
only. This is followed by a description of the class of P-solvable loops and its
reach and limitations in Section~\ref{sec:psolvable}. In
Section~\ref{sec:extension} we present our main contribution, an extension of
P-solvable loops by reasoning about hypergeometric sequences and we derive the
necessary theoretical and algorithmical results to offer fully automated
polynomial invariant generation therein. We conclude the paper with a
presentation of our implementation in the Mathematica package Aligator in
Section~\ref{sec:aligator} and a summary of possible future research directions
in Section~\ref{sec:conclusion}.

\subsection{Related Work}\label{sec:related}

Many classical data flow analysis problems, such as constant propagation and
finding definite equalities among program variables, can be seen as problems
about polynomial identities expressing loop invariants.  In \cite{olm, san} a
method built upon linear and polynomial algebra is developed for computing
polynomial equalities of a bounded degree.  A related approach was also proposed
by~\cite{Carbonell07a} using abstract interpretation. Abstract interpretation is
also used in~\cite{farzan,oliveira} for computing polynomial invariants of
programs whose assignments can be described by C-finite recurrences. In our work
we do not rely on abstract interpretation but use algebraic reasoning about
holonomic sequences. For program loops with assignments only, our technique can
handle programs with more complex arithmetic than the previously mentioned
methods. Our work is currently restricted though to single-path loops.

Without an a priori fixed polynomial degree, in~\cite{kapur} the polynomial
invariant ideal is approximated by a fixed point procedure based on polynomial
algebra and abstract interpretation.  In~\cite{laura}, the author defines the
notion of P-solvable loops which strictly generalizes the programming model
of~\cite{kapur}. Given a P-solvable loop with assignments and nested
conditionals, the results in~\cite{laura} yield an automatic approach for
computing all polynomial loop invariants.  Our work extends~\cite{laura,kapur}
in new ways: it handles a richer class of P-solvable loops where multiplication
with the loop counter is allowed. Our technique relies on manipulating
hypergeometric sequences and relaxes the algebraic restrictions
of~\cite{laura,kapur} on program operations. To the best of our knowledge, no
other method is able to derive polynomial invariants for extended P-solvable
loops.  Unlike~\cite{laura,kapur}, we however only treat loops with assignments;
that is, invariants for extended P-solvable loops with conditionals are not yet
treated by our approach.

\section{Preliminaries}\label{sec:preliminaries}
In this section we give a brief overview of the algebra of linear ordinary
recurrence operators as well as C-finite and hypergeometric sequences that we use
further on. We also describe our programming model in detail.

\subsection{Recurrence Operators and Holonomic Sequences}
\label{sec:ore}
Let $\set K$ be a computable field of characteristic
zero. 

The algebra of linear ordinary recurrence operators in one
variable will serve as the algebraic foundation to deal with recurrence
equations. For details on general Ore algebras, see~\cite{bronstein,ore}. 

\begin{definition}
\label{def:ore}
Let $\set K(x)[S]$ be the set of univariate polynomials in the variable $S$ over
the set of rational functions $\set K(x)$ in $x$ and let
$\sigma\colon \set K(x)\rightarrow \set K(x)$ be the forward shift
operator in $x$, i.e.\
$\sigma(r(x))=r(x+1)$ for $r(x)\in\set K(x)$. We define the \textit{Ore
  polynomial ring of ordinary recurrence operators} $(\set K(x)[S],+,\cdot)$
with component-wise addition and the unique distributive and associative
extension of the multiplication rule
\begin{displaymath}
 S a = \sigma(a)S\text{\qquad for all }a\in\set K(x),
\end{displaymath}
to arbitrary polynomials in $\set K(x)[S]$. To clearly distinguish this ring
from the commutative polynomial ring over $\set K(x)$, we denote it by
$\set K(x)[S;\sigma,0]$. The {\it order} of an operator $L\in \set K(x)[S;\sigma,0]$
is its degree in $S$.
\end{definition}
Without loss of generality, we assume that the leading coefficient of any
operator $L\in\set K(x)[S;\sigma,0]$ is equal to 1. Otherwise, we can divide by
the leading coefficient of $L$ from the left. $\set K(x)[S;\sigma,0]$ is a right
Euclidean domain, i.e.\ we have the notion of the greatest common right divisor
and the least common left multiple of operators and we are able to determine
both algorithmically. Consequently, $\set K(x)[S;\sigma,0]$ is a principal left
ideal domain and every left ideal is generated by the greatest common right
divisor of a given set of generators.

Consider the ring $\set K^\set N$ of all sequences in $\set K$ with
component-wise addition and the Hadamard product (i.e.\ component-wise product)
as multiplication. We follow~\cite{zeilberger} in identifying sequences as equal
if they only differ in finitely many terms. This will prove beneficial in two
ways. Firstly, it allows us to define the action of operators on sequences in a
natural way. Secondly, disregarding finitely many starting values makes it
possible to identify unnecessary loop variables, whose values are eventually
equal to the values of another variable, and therefore can be computed outside
of any while loop. Let $\sim$ be the equivalence relation on $\set K^\set N$
defined by
\[s\sim t:\Leftrightarrow s-t\text{ has finitely many non-zero elements}.\] We
then set $\mathcal S$ to be the quotient ring $\set K^\set N/\sim$.
Subsequently, it will not be necessary to distinguish between
$t\in\set K^\set N$ and $\pi(t)\in\mathcal S$, where
$\pi:\set K^\set N\rightarrow \mathcal{S}$ is the canonical homomorphism. The
field~$\set K$ can be embedded in $\mathcal S$ via the map
$c\mapsto (c)_{n\in\set N}$.  The action of an operator in
$\set K(x)[S;\sigma,0]$ on an element in $\mathcal S$ is defined by the map
\begin{align*}
  & \tau: \set K(x)[S;\sigma,0]\times \mathcal S \rightarrow \mathcal S\notag \\
  & \tau(L(S,x),t)(n)=\tau\biggl(\sum_{i=0}^{d}l_i(x)S^i,t\biggr)(n) :=
  \sum_{i=0}^{d}l_i(n)t(n+i),
\end{align*}
where the evaluation is well defined for all $n\geq n_0$ for some
$n_0\in\set N$, and we set $L(t):=\tau(L,t)\in\mathcal S$.  If $L(t)\equiv0$,
then we say that $L$ is an \textit{annihilator} of $t$ ($L$ \textit{annihilates}
$t$) and $t$ is a \textit{solution} of $L(t)=0$. A sequence that is annihilated
by a non-zero operator in $\set K(x)[S;\sigma,0]$ is called {\it holonomic
  sequence}. For a given sequence $t$, the set of all its annihilators forms a
left ideal in $\set K(x)[S;\sigma,0]$. We call it the {\it annihilator ideal} of
$t$ and denote it by $\operatorname{ann}(t)$.

\begin{example}
\label{ex:poly}
Let $p(x)$ be a polynomial in $\set K[x]$. The polynomial sequence
$(p(n))_{n\in\set N}$ is annihilated by the operator \[L_1=S-\frac{p(x+1)}{p(x)}.\]
$L_1$ is a generator of the annihilator ideal of $p$.  Set $\Delta:=S-1$. Then
$\tilde{p}=\Delta(p)$ is again a polynomial sequence with
$\deg(\tilde{p})<\deg(p)$. It follows that $L_2=\Delta^{\deg(p)+1}$ is another
annihilator of $p$ in $\set K(x)[S;\sigma,0]$ and its coefficients are
independent of $x$. Since $L_1$ generates $\operatorname{ann}(p)$, there exists
an operator $Q$ with $L_2=QL_1$.
\end{example}

In our work, we focus on two different special kinds of holonomic sequences:\goodbreak
\begin{definition}
\label{def:hg}
Let $t\in \mathcal{S}$. Then 
\begin{itemize}
\item $t$ is called \textit{C-finite} if it is annihilated by an operator in $\set
  K(x)[S;\sigma,0]$ with only constant coefficients. $(l_i\in \set K$)
\item $t$ is called \textit{hypergeometric} if it is
annihilated by an order 1 operator in $\set K(x)[S;\sigma,0]$.
\end{itemize}
\end{definition}

\begin{example}
\label{ex:seq}
We give some examples of commonly encountered sequences.
 \begin{itemize}
 \item As was shown in Example~\ref{ex:poly}, polynomial sequences are both,
   C-finite and hypergeometric.
 \item Rational function sequences $(r(n))_{n\in\set N}$, $r\in\set
   K(x)\setminus\set K[x]$, are hypergeometric but not C-finite.
 \item The factorial sequence $(n!)_{n\in\set N}$ is hypergeometric but not
   C-finite.
 \item The Fibonacci sequence $(f(n))_{n\in\set N}$
   with
   \[f(n)=\frac1{\sqrt{5}}\left(\left(\frac{1+\sqrt{5}}2\right)^{\mathrlap{n}} -
       \left(\frac{1-\sqrt{5}}2\right)^n\right),\]
   is C-finite but not hypergeometric.
 \item The sequence of harmonic numbers $(h(n))_{n\in\set N}$
   with \[h(n)=\sum_{i=1}^n{\frac1i},\] is neither hypergeometric nor C-finite.
 \end{itemize}
\end{example}

In a sufficiently large algebraic field extension $\overline{\set K}/\set K$, every
C-finite sequence $(c(n))_{n\in\set N}$ can be uniquely written (up to
reordering) in the form
\[c(n)=p_1(n)\theta_1^n +p_2(n)\theta_2^n+\dots+p_s(n)\theta_s^n,\]
for some $s\in\set N$ and $p_i\in\set K[x], \theta_i\in\overline{\set K}$ for
$i=1,\dots,s$ with $\theta_i\neq\theta_j$ for $i\neq j$.  For any
$r\in\set K(x)$ and $n\in\set N$, $r(x)^{\underline n}$ is defined as
$\prod_{i=0}^{n-1}r(x-i)$.  Then every hypergeometric sequence
$(h(n))_{n\in\set N}$ can be uniquely written (up to reordering) in the form
\[h(n) = \theta^n r(n)((n+\zeta_1)^{\underline n})^{k_1}((n+\zeta_2)^{\underline
    n})^{k_2}\cdots ((n+\zeta_\ell)^{\underline n})^{k_\ell},\]
for some $\ell\in\set N$, $r(x)\in\set K(x)$, $\theta\in\overline{\set K}$,
$\zeta_i\in\overline{\set K}$ and $k_i\in\set Z$ for $i=1,\dots,\ell$, and the
difference $\zeta_i-\zeta_j$ is not an integer for $i\neq j$. From these closed
forms it is immediate that finite sums and products of C-finite sequences are
again C-finite and finite products of hypergeometric sequences are again
hypergeometric. Sums of hypergeometric sequences are not necessarily
hypergeometric, see Lemma~\ref{lem:hgsum}. Subsequently, we will assume that
$\set K$ is large enough so that all occurring C-finite and hypergeometric
sequences have a closed form representation in $\set K$.

For more details on C-finite and
hypergeometric sequences, as well as proofs for the facts given in this section,
see~\cite{kauers}.

For functions $f_1,\dots,f_m: \set U\rightarrow \set K$ with
$\set N\subset\set U\subset\set K$ that are algebraically independent over
$\set K$, we distinguish between the polynomial ring $\set K[f_1,\dots,f_m]$,
where $f_1,\dots,f_m$ are used as variables, and the ring
$\set K[f_1(n),\dots,f_m(n)]\subset \mathcal S$ of all sequences
$(t(n))_{n\in\set N}$ of the form $t(n)=p(f_1(n),\dots,f_m(n))$ with
$p\in\set K[f_1,\dots,f_m]$. This distinction is important, as e.g.\ the
function $\sin(x\cdot\pi)$ is algebraically independent over $\set K$, but the
sequence $(\sin(n\cdot\pi))_{n\in\set N}=(0,0,0,\dots)$ is not, and thus
$\set K[\sin(n\cdot\pi)]$ is isomorphic to $\set K$, but
$\set K[\sin(x\cdot\pi)]$ is not.

\begin{remark}
  In the context of this paper, since the operators in question emerge from
  program loops, we can safely assume that the rational function coefficients of
  any operator do not have poles in $\set N$. Otherwise, a division by zero
  error would occur for some program input. 
\end{remark}

\subsection{Programming Model}
\label{sec:model}

We consider a simple programming model of single-path loops with rational
function assignments. That is, nested loops and/or loops with
conditionals are not yet handled in our work. Our programming model is
thus given by the following loop pattern, written in a C-like syntax: 

\begin{equation}\label{eq:loop}
    \begin{tabular}{ll}
      \WHILE\ $pred(v_1,\dots,v_m)$\ \DO\\
      \quad $\phantom{v_m}\mathllap{v_1\;}\ass f_1(v_1,\dots,v_m)$;\\
      \quad\quad\vdots\\
      \quad $v_m\ass f_m(v_1,\dots,v_m);$\\
      \OD\\
    \end{tabular}
\end{equation}

\noindent where $v_1,\dots,v_m$ are (scalar) variables with values from
$\set K$, the $f_i$ are rational functions over $\set K$ in $m$ variables and
$pred$ is a a Boolean formula (loop condition) over $v_1,\dots,v_m$. In our
approach however we ignore loop conditions and treat program loops as
non-deterministic programs. In~\cite{olm}, it is shown that the set of all
affine equality invariants is not computable if the programming model includes
affine equality tests/conditions. With this consideration, our programming model
from~\eqref{eq:loop} becomes:

\begin{equation}\label{eq:prgModel}
    \begin{tabular}{ll}
      \WHILE\ true\ \DO\\
      \quad\quad\vdots\\
      \OD\\
    \end{tabular}
\end{equation}

\noindent Due to particular importance in our reasoning, we suppose that there
is always a variable $n$ denoting the loop iteration counter. The initial value
of $n$ will always be $n=0$ and $n$ will be incremented by $1$ at the end of
each iteration.

Each program variable gives rise to a sequence $(v_i(n))_{n\in\set N}$. For a
program variable $v$, we allow ourselves to abuse the notation and also use the
identifier $v$ as a variable in polynomial rings as well as an identifier for
the sequence $(v(n))_{n\in\set N}$. 

A polynomial loop invariant is a non-zero polynomial $p$ over $\set K$ in $m$
variables such that $p(v_1(n),\dots,v_m(n))=0$ for all $n$. As observed
in~\cite{kapur,laura}, the set of all polynomial invariants forms a polynomial
ideal in $\set K[v_1,\dots,v_m]$, called the {\it polynomial invariant ideal}
and is denoted by $I(v_1,\dots,v_m)$. For a subset
$\{\tilde{v}_1,\dots,\tilde{v}_k\}\subset\{v_1,\dots,v_m\}$, we define
\[I(\tilde{v}_1,\dots,\tilde{v}_k) = I(v_1,\dots,v_m)\cap\set
  K[\tilde{v}_1,\dots,\tilde{v}_k].\]

In general, polynomial loop invariants depend on the initial values of program
variables. To simplify the presentation, we fix $\set K$ to be
\[\set K=\set F(v_{1,0},\dots,v_{1,k},v_{2,0},\dots,v_{m,\ell}),\]
for a computable field $\set F$ of characteristic zero that allows us to
represent all occurring C-finite and hypergeometric sequences in closed form, and
sufficiently many variables $v_{1,0},\dots,v_{m,\ell}$ that represent the
initial values of the program variables $v_1,\dots,v_m$.

\section{Polynomial Invariants for P-Solvable Loops}
\label{sec:psolvable}
We now turn our attention to the class of P-solvable loops introduced
in~\cite{laura} that allows for computing all polynomial loop invariants..

\begin{definition}
\label{def:psolvable}
An imperative loop with assignment statements only is called \textit{P-solvable}
if the sequence of each recursively changed program variable $v$ is C-finite and
the ideal of all polynomial invariants over $\set K$ is not the zero ideal.
\end{definition}

\begin{example}
  In~\cite{laura}, it is shown that the Euclidean algorithm is P-solvable. Given
  the program:

 \vspace{0.2cm}
    \begin{tabular}{ll}
     \WHILE\ $y \leq rem$\ \DO\\
      \quad $rem \ass rem - y$; \\
      \quad $quo \ass quo + 1$; \\
      \OD
    \end{tabular}

    \vspace{0.12cm} \noindent The  ideal
    of polynomial loop invariants is shown to be 
    \[I(quo,rem,x,y)=\langle rem + quo\cdot y - y\cdot quo(0) -
    rem(0)\rangle.\] With $quo(0) = 0$ and $rem(0) = x$, this gives
    $\langle rem + quo\cdot y - x\rangle$.
\end{example}

While P-solvable loops cover a wide class of program loops, there are several
significant cases which do not fall into this class. Notably, multiplication
with the loop counter $n$ will generally result in loops that are not
P-solvable.

\begin{example}
\label{ex:mainex1}
Consider the following loop with relevant loop variables $a,b,c,d$. The
variables $t_1,t_2$ are temporary variables used to access previous values of
$a$. Along with the loop counter $n$, we will not take them into consideration for
the loop invariants in this example.

\vspace{0.2cm}
    \begin{tabular}{ll}
      \WHILE\ true\ \DO\\
      \quad $t_1\ass t_2;\quad t_2\ass a$;\\
      \quad $a\ass  5(n+2)\cdot t_2 + 6\cdot (n^2+3\cdot n+2)\cdot t_1$;\\ 
      \quad $b\ass  2\cdot b$;\\ 
      \quad $c\ass  3\cdot (n+2)\cdot c$;\\ 
      \quad $d\ass  (n+2)\cdot d$;\\ 
      \quad $n\ass n+1$;\\
      \OD\\
    \end{tabular}

\vspace{0.12cm}
\noindent The program then satisfies the following system of recurrences:
\begin{equation*}
\begin{cases}
a(n+2)-5(n+2)\cdot a(n+1)-6(n^2+3n+2)\cdot a(n)=0\\
b(n+1)-2\cdot b(n)=0\\
c(n+1)-3(n+1)\cdot c(n)=0\\
d(n+1)-(n+1)\cdot d(n)=0.
\end{cases}
\end{equation*}
This loop is not P-solvable as, for example, the variable $c$ is
updated by a sequence that is not C-finite (due to the multiplication
between the program variables $n$ and $c$). To the best of our
knowledge, none of the  existing invariant generation techniques is
able to to compute polynomial invariants for this loop. 
In the next section, we extend the class of P-solvable loops, covering also
programs as the one above, and introduce an automated approach to
derive all polynomial invariants of such loops. 
\end{example}

\section{Extension of P-Solvable Loops}
\label{sec:extension}

\subsection{Definition of Extended P-Solvable Loops}
\label{sec:defext}

Consider the sequences $(v_1(n))_{n\in\set N},\dots, (v_m(n))_{n\in\set N}$ with
values in $\set K$ given by
\begin{equation}
\label{eq:psolvable}
v_i(n)=\sum_{k\in\set
  Z^\ell}p_{i,k}(n,\theta_1^n,\dots,\theta_s^n)((n+\zeta_1)^{\underline n})^{k_1}\cdots
((n+\zeta_\ell)^{\underline n})^{k_\ell}
\end{equation}
where $s,\ell\in\set N$, the $p_{i,k}$ are polynomials in
$\set K(x)[y_1,\dots,y_s]$, not identically zero for finitely many
$k\in\set Z^\ell$, and the $\theta_i$ and $\zeta_j$ are elements of $\set K$ for
$i=1,\dots,s$, $j=1,\dots,\ell$ with $\theta_i\neq\theta_j$ and
$\zeta_i-\zeta_j\notin\set Z$ for $i\neq j$.

In particular, this class of sequences comprises C-finite sequences as well as
hypergeometric sequences and Hadamard products of C-finite and hypergeometric
sequences, which could not be handled in automated invariant generation
before. We give an extension of Definition~\ref{def:psolvable} based on this
class of sequences

\begin{definition}
\label{def:extended}
An imperative loop with assignment statements only is called \textit{extended
  P-solvable} if the sequence of each recursively changed program variable $v$
is of the form~\eqref{eq:psolvable}.
\end{definition}

Note that in Definition~\ref{def:extended}, we drop the requirement of
Definition~\ref{def:psolvable} that the ideal of algebraic relations is not the
zero ideal. This change is just for convenience.

While it is obvious that the inclusion of hypergeometric terms in extended
P-solvable loops allows assignments of the form $v\ass r(n)v$, where $r$ is a
rational function in $\set K[x]$, it also allows assignments that turn into
higher order recurrences, as illustrated in Example~\ref{ex:mainex2}. It also
allows for assignments of the form $v_2\ass r(v_1)v_2$, with $r\in\set K(x)$, as
long as the closed form of $v_1$ is a rational function in $n$.
 
\subsection{Detecting Extended P-Solvable Loops}
In order to employ the ideas we develop in Section~\ref{sec:relations} for
finding algebraic relations in extended P-solvable loops, we have to be able to
detect sequences of the form~\eqref{eq:psolvable}.  This means, given a
recurrence operator $R$ of order $d$ and starting values $s_0,\dots,s_{d-1}$,
compute, if possible, $p_k$,$\theta_i$ and $\zeta_j$ as in~\eqref{eq:psolvable}
such that $v$ is a solution of $R(v)=0$ with $v(n)=s_n$ for
$n\in\{0,\dots,d-1\}$.  We can write $v$ as a sum of hypergeometric sequences:
\begin{alignat*}2
  &v(n)  = h_1(n) + \dots + h_w(n),\text{ where}\\
  & h_i(n) =q_i(n)\tilde{\theta}_i^n((n+\zeta_1)^{\underline n})^{k_{i,1}}\cdots
  ((n+\zeta_\ell)^{\underline n})^{k_{i,\ell}},
\end{alignat*}
with $q_i\in\set K(x)$, $\tilde{\theta}_i\in\set K$, and $k_i\in\set Z^\ell$.
Note that we use $\tilde{\theta}_i$ instead of $\theta_i$ since the exponential
sequence for each summand can be a product of several $\theta_i^n$.  We can
assume without loss of generality that the $h_i$ are linearly independent over
$\set K(n)$. In fact, if $h_1(n)= r_2(n)h_2(n)+\dots+r_w(n)h_w(n)$, we can set
$\tilde{h}_1=(1+r_2)h_2,\dots,\tilde{h}_{w-1}=(1+r_w)h_w$ and get
$v(n)=\tilde{h}_1(n) + \dots + \tilde{h}_{w-1}(n)$. Let $L$ be the least common
left multiple of the first order operators $L_1,\dots,L_w$ that annihilate
$h_1,\dots,h_w$ respectively in the Ore algebra $\set K(x)[S;\sigma,0]$ and let
$G$ be a generator of $\operatorname{ann}(v)$. We show that $G$ and $L$ are
equal. (Note that we required all operators to have leading coefficient $1$.)

\noindent By right division with remainder, we can write $G$ as
\begin{alignat*}2
G =&\;Q_1L_1+r_1\\
 =&\;Q_2L_2+r_2\\
 &\vdots\\
 =&\;Q_wL_w+r_w,
\end{alignat*}
with $Q_1,\dots,Q_w\in\set K(x)[S;\sigma,0]$ and some $r_1,\dots,r_w\in
\set K(x)$. We then get
\[0=G(v)=G(h_1+\dots+h_w)=G(h_1)+\dots+G(h_w)=r_1h_1+\dots r_wh_w.\]
Since the $h_i$ are linearly independent, we have $r_1=\dots=r_w=0$, and
so, $L_1,\dots,L_w$ are right factors of $G$. This proves the claim.

Since every annihilator of $v$ is a multiple of $G$ and therefore also an
annihilator of $h_i$, we can use Petkov\v{s}ek's algorithm~\cite{pet} to
determine $p_k$,$\theta_i$ and $\zeta_j$ as in~\eqref{eq:psolvable}. More
precisely, given an operator $R\in\set K(x)[S;\sigma,0]$ of order $d$ and
starting values $s_0,\dots,s_{d-1}$, we compute $v$ as in~\eqref{eq:psolvable}
such that $R(v)=0$ (if possible), by computing all hypergeometric solutions of
$R$. This gives $\theta_i,\zeta_i$ and $p_i$, linearly dependent on parameters
$c_1,\dots,c_w$. Next, we solve the linear system $v(i)=s_i$ in terms of
$c_i$. Any solution then gives rise to a sequence $(v(n))_{n\in\set N}$ with the
desired properties.

\begin{example}
\label{ex:mainex2}
For the recurrence for $a$ in Example~\ref{ex:mainex1}, we compute two
hypergeometric solutions using Petkov\v{s}ek's algorithm:
\[h_1=(-1)^nn!,\quad h_2=6^nn!\]
Thus, we get 
\[a(n)=(k_1(-1)^n+k_26^n)n!\]
with the relations $a(0)=k_1+k_2$ and $a(1)=6k_2-k_1$ stemming from the starting
values of $a$.  Since $b,c,d$ are given by first order recurrences, their closed
forms can be easily computed:
\[b(n)=2^nb(0),\quad c(n)=3^nn!c(0),\quad d(n)=n!d(0).\]
It follows that the program loop given in Example~\ref{ex:mainex1} is extended
P-solvable.
\end{example}
 \subsection{The Ideal of Algebraic Relations}
\label{sec:relations}
We now turn to the problem of, given sequences $v_1,\dots,v_m$ as
in~\eqref{eq:psolvable}, how to compute a basis for the ideal
$I(v_1,\dots,v_m)$ of all algebraic relations among the $v_i$. We proceed by
identifying the terms $(n+\zeta_i)^{\underline n}$ that are algebraically
independent over $\set K(n,\theta_1^n,\dots,\theta_s^n)$. For this, we use basic
properties of sums and products of hypergeometric terms. First, we state a
necessary condition for a finite sum of hypergeometric terms to be again
hypergeometric.
\begin{lemma}
\label{lem:hgsum}
  Let $h_1,\dots,h_w$ be hypergeometric sequences. If the
  sum $h_1+\dots+h_w$ is hypergeometric, then there exist integers
  $i,j\in\{1,\dots,w\}$, $i\neq j$, and a rational function $r(x)\in\set K(x)$ such that
  $h_i(n)=r(n)h_j(n)$.
\end{lemma}
\begin{proof}
  We prove the claim by induction on $w$. For the case $w=1$, there is nothing
  to show. Now suppose the claim holds for some $(w-1)\in\set N^*$. There
  is a rational function $r_h(x)\in\set K(x)$ such that
 \[\sum_{i=1}^w h_i(n+1)=r_h(n)\sum_{i=1}^w h_i(n).\]
 Let $r_i\in\set K(x)$ be such that $h_i(n+1)=r_i(n)h_i(n)$. We then get
 \begin{equation}
 \label{eq:hgsum}
 \sum_{i=1}^w (r_i(n)-r_h(n))h_i(n)=0.
 \end{equation}
 We first treat the case in which for all $i$, $(r_i(x)-r_h(x))$ is not zero.
 Then, bringing $(r_w(n)-r_h(n))h_w(n)$ in~\eqref{eq:hgsum} to the other
 side yields
 \[\sum_{i=1}^{w-1}(r_i(n)-r_h(n))h_i(n)=(r_w(n)-r_h(n))h_w(n).\]
 The sequence $(r_w(n)-r_h(n))h_w(n)$ is hypergeometric, and by the induction
 hypothesis it follows that there are $i,j$ and a rational function $\tilde{r}$
 with $(r_i(n)-r_h(n))h_i(n)= \tilde{r}(n)(r_j(n)-r_h(n))h_j(n)$.
 Dividing by $r_i(n)-r_h(n)$ proves the claim.
 For the case that there is an $i$ with $(r_i(x)-r_h(x))=0$, the left hand side
 of~\eqref{eq:hgsum} is a sum of fewer than~$w$ hypergeometric terms and the
 right hand side is hypergeometric. The induction hypothesis then again yields
 suitable $i,j$ and $r(x)$.
\end{proof}

\begin{example}
  The sums $2n!+(n+3)!$ and $n! + (n+\frac12)^{\underline n} -n!$ are
  hypergeometric, whereas $1+n!$ is not.
\end{example}

The next lemma gives a characterization of when the quotient of two
hypergeometric sequences is a rational function sequence. Together with
Lemma~\ref{lem:hgsum}, this then will yield the algebraic independence of
certain hypergeometric sequences in Lemma~\ref{lem:trans}.

\begin{lemma}
\label{lem:rat}
Let $\zeta_1,\dots,\zeta_\ell\in\set K$ be such that for all $i,j=1,\dots,\ell$
with $i\neq j$, we have $\zeta_i-\zeta_j\notin\set Z$. Then for
$k_1,\dots,k_\ell\in\set N$, $c_1,\dots,c_\ell\in\set N$, and
$\theta_1,\theta_2\in\set K$, there is a rational
function $r(x)\in\set K(x)$ such that
\begin{alignat*}1 & \theta_1^n\cdot((n-\zeta_1)^{\underline n})^{k_1}\cdots
  ((n-\zeta_\ell)^{\underline n})^{k_\ell}
  = \\
  &\quad r(n)\cdot\theta_2^n\cdot((n-\zeta_1)^{\underline n})^{c_1}\cdots
  ((n-\zeta_\ell)^{\underline n})^{c_\ell},
\end{alignat*}
  if and only if $\theta_1=\theta_2$ and $(k_1,\dots,k_\ell)=(c_1,\dots,c_\ell)$.
\end{lemma}
\begin{proof}
  If $\theta_1=\theta_2$ and $(k_1,\dots,k_\ell)=(c_1,\dots,c_\ell)$, then we
  can set $r(x)=1$. For the other direction, we have
  \[\underbrace{\left(\frac{\theta_1}{\theta_2}\right)^{\mathrlap{n}\hspace{1px}}
      ((n-\zeta_1)^{\underline n})^{k_1-c_1}\cdots 
      ((n-\zeta_\ell)^{\underline n})^{k_\ell-c_\ell}}_{\text{hypergeometric}}= r(n).\]
  A hypergeometric term $h$ is a rational function if and only if its shift
  quotient $h(x+1)/h(x)$ can be written in the
  form \[q(x)=\frac{g(x)f(x+1)}{g(x+1)f(x)},\]
  with $f,g\in\set K[x]$. Therefore, for any root in the numerator of $q(x)$
  there is a root in integer distance in the denominator of $q(x)$, which, by
  the condition on the $\zeta_i$, is not possible if $\theta_1\neq\theta_2$
  or $(k_1,\dots,k_\ell)\neq(c_1,\dots,c_\ell)$
\end{proof}
\begin{lemma}
\label{lem:trans}
Let $\theta_1,\dots,\theta_s\in\set K$ and $\zeta_1,\dots,\zeta_\ell\in\set K$.
The sequences
$(n+\zeta_1)^{\underline n},(n+\zeta_2)^{\underline
  n},\dots,(n+\zeta_\ell)^{\underline n}$
are algebraically independent over $\set K(n,\theta_1^n,\dots,\theta_s^n)$ if
and only if there are no $i,j\in\{1,\dots,\ell\}$, $i\neq j$ such that
$\zeta_i-\zeta_j\in\set Z$.
\end{lemma}
\begin{proof}
  If there are $i,j\in\{1,\dots,\ell\}$, $i\neq j$ with
  $\zeta_i-\zeta_j=k\in\set Z$, then we get the algebraic relation
  \[(n+\zeta_i)^{\underline n}\cdot\prod_{w=1}^{k}(\zeta_j-w) =
    (n+\zeta_j)^{\underline n}\cdot\prod_{w=1}^{k}(n+w+\zeta_j).\]
  Conversely, let $p$ be a nonzero polynomial over
  $\set K(n,\theta_1^n,\dots,\theta_s^n)$ in~$\ell$ variables.  We can write
  $\operatorname{denominator}(p)\cdot p((n+\zeta_1)^{\underline n},\dots,(n+\zeta_\ell)^{\underline n})$ as a sum
  of the form
\[\sum_{i\in\set N,k\in\set
    Z^\ell}p_{i,k}(n)\tilde{\theta}_i^n((n+\zeta_1)^{\underline
    n})^{k_1}\cdots ((n+\zeta_\ell)^{\underline n})^{k_\ell}\]
Assume that
$p((n+\zeta_1)^{\underline n},\dots,(n+\zeta_\ell)^{\underline n})=0.$ Then, by
Lemma~\ref{lem:hgsum}, there have to be terms $(i,k),(j,c)\in\set N\times\set Z^\ell$, $(i,k)\neq (j,c)$
and a rational function $r(x)\in\set K(x)$ with
\begin{alignat*}1
& p_{i,k}(n)\tilde{\theta}_i^n((n-\zeta_1)^{\underline n})^{k_1}\cdots 
  ((n-\zeta_\ell)^{\underline n})^{k_\ell} =\\
&\quad r(n)p_{j,c}(n)\tilde{\theta}_j^n((n-\zeta_1)^{\underline
  n})^{c_1} \cdots ((n-\zeta_\ell)^{\underline n})^{c_\ell},
\end{alignat*}
By Lemma~\ref{lem:rat}, this can only be the case if there are
$\zeta_i\neq\zeta_j$ in integer distance, which contradicts the condition on the
$\zeta_i$.
\end{proof}

\begin{example}
 Let $h_1,h_2,h_3$ be hypergeometric sequences given by $h_1(0)=h_2(0)=h_3(0)=1$
 and 
 \begin{align*}
 & h_1(n+1)=(n^2+\frac32 n + \frac12)h_1(n),\;
   h_2(n+1)=(n+1)h_2(n),\\
 & h_3(n+1)=\frac{2n^3 + 9n^2 + 10n + 3}{2n+4}h_3(n).
\end{align*}
 The closed forms then are
 \begin{align*}
   &h_1(n) =\prod_{i=0}^n(i^2+\frac32 i + \frac12) = \prod_{i=0}^n(i+1)(i+\frac12) =
     (n+1)^{\underline{n}}(n+\frac12)^{\underline{n}},\\
   &h_2(n) =\prod_{i=0}^n(i+1) = (n+1)^{\underline{n}},\\
   &h_3(n) =\prod_{i=0}^n\frac{2i^3 + 9i^2 + 10i + 3}{2i+4} = \prod_{i=0}^n
     \frac{(i+1)(i+\frac12)(2(i+1)+4)}{2i+4} =\\
   &\hspace{30px}(2n+4)(n+1)^{\underline{n}}(n+\frac12)^{\underline{n}}.
\end{align*}
 From Lemma~\ref{lem:trans} it follows that $h_1,h_2$ are algebraically
 independent over $\set K$, but $h_1,h_3$ are not.
\end{example}

Lemma~\ref{lem:trans} allows us to represent the sequences arising in extended
P-solvable loops as rational function sequences over the field
$\set K(n,\theta_1^n,\dots,\theta_s^n)$ as follows: Let $v_1,\dots,v_m$ be of
the form~\eqref{eq:psolvable} and let
$\tilde{Z}=\{\tilde{\zeta}_1,\dots\tilde{\zeta}_k\}$ be a subset of
$Z=\{\zeta_1,\dots,\zeta_\ell\}$ such that there are no $i,j=1,\dots,k$,
$i\neq j$, with $\tilde{\zeta}_i-\tilde{\zeta}_j\in\set Z$ and for each
$\zeta\in Z\setminus\tilde{Z}$ there exists an $i$ such that
$\tilde{\zeta}_i-\zeta\in\set Z$. Let
$z_1,\dots,z_\ell\in\set K[x,y_1,\dots,y_k]$ be such that
\[z_i(n,(n-\tilde{\zeta}_1)^{\underline n},\dots,(n-\tilde{\zeta}_k)^{\underline
    n}) = (n-\zeta_i)^{\underline n},\]
for all $n\in\set N$ and $i=1,\dots,\ell$.  Then there exist
$k_1,\dots,k_m\in\set Z^\ell$ with
\begin{equation*}
\begin{alignedat}2 v_i(n)= &\sum_{j\in\set
    Z}p_{i,j}(n,\theta_1^n,\dots,\theta_s^n)\cdot{}\\
  &\prod_{\mathclap{1\leq w \leq \ell}}z_w(n,(n-\tilde{\zeta}_1)^{\underline
    n},\dots,(n-\tilde{\zeta}_k)^{\underline n})^{k_{i,w}}.
\end{alignedat}
\end{equation*}
Substituting variables $v_i$ for $v_i(n)$, $h_i$ for
$(n-\tilde{\zeta}_i)^{\underline n}$, $e_i$ for $\theta_i^n$ and $x$ for $n$
then gives
\[\left[v_i = r_i(x,e_1,\dots,e_s,h_1,\dots,h_k)\right]_{\begin{subarray}{l}
      v_i\rightarrow v_i(n), h_i\rightarrow (n-\tilde{\zeta}_i)^{\underline
        n},\\e_i\rightarrow \theta_i^n, x\rightarrow n\end{subarray}}\]
where $r_{i}$ is a rational function over $\set K$ in $1+s+k$ variables. We
now can compute the ideal of all algebraic dependencies among the program
variables of a P-solvable loop as the ideal of algebraic relations among
rational functions.

\begin{proposition}
\label{prop:mainprop}
Let $(v_1(n))_{n\in\set N},\dots,(v_m(n))_{n\in\set N}$ be sequences of the
form~\eqref{eq:psolvable} and consider the corresponding rational functions
$r_1,\dots,r_m$ in $\set K(x,e_1,\dots,e_s,h_1,\dots,h_k)$ as above. For
each $i=1,\dots,m$, write $r_i=f_i/g_i$ with coprime polynomials $f_i,g_i$
over $\set K$. Denote by $I(\theta_1^n,\dots,\theta_s^n)$ the ideal of algebraic
relations among $\theta_1^n,\dots,\theta_s^n$ in $\set K[e_1,\dots,e_s]$. Then
the ideal of algebraic relations among the sequences
$(v_1(n))_{n\in\set N},\dots,(v_m(n))_{n\in\set N}$ in $\set K[v_1\dots,v_m]$
is given by
\begin{alignat*}2
  I(v_1,\dots,v_m) ={}&(I(\theta_1^n,\dots,\theta_s^n) +&{}\\
  &\langle g_1v_1-f_1,\dots,g_m v_m-f_m \rangle )\cap \set
  K[v_1,\dots,v_m].
\end{alignat*}
\end{proposition}
\begin{proof}
  The proposition follows immediately from the fact that the ideal of algebraic
  dependencies among a set of rational functions
  \[\frac{r_1(x_1,\dots,x_k)}{d_1(x_1,\dots,x_k)},\dots,
    \frac{r_m(x_1,\dots,x_k)}{d_m(x_1,\dots,x_k)},\]
  in the polynomial ring $\set K[y_1,\dots,y_m]$ is given by
\begin{alignat*}2
& \langle
  d_1(x_1,\dots,x_k)y_1-r_1(x_1,\dots,x_k),\dots,\\
&\phantom{\langle}d_m(x_1,\dots,x_k)y_m-r_m(x_1,\dots,x_k)\rangle\cap
  \set K[y_1,\dots,y_m],
\end{alignat*}
and that by Lemma~\ref{lem:trans} there are no algebraic relations over the
field $\set K(n,\theta_1^n,\dots,\theta_s^n)$ among the terms
$(n-\tilde{\zeta}_i)^{\underline n}$ with $\tilde{\zeta}_i$ as above for
$i=1,\dots,k$.
\end{proof}

\begin{example}
  We compute the ideal of algebraic relations among $a,b,c,d$ given in
  Example~\ref{ex:mainex1}. First, we compute the ideal of algebraic relations
  among $(-1)^n,2^n, 3^n$ and $6^n$ with corresponding variables
  $e_{-1},e_2,e_3,e_6$.  We get
  \[I((-1)^n,2^n,3^n,6^n)=\langle e_{-1}^2 - 1, e_2e_3 - e_6\rangle.\]
  Now we can compute the ideal of algebraic relations among $a,b,c,d$ by adding
  the relations
  $a-(k_1e_{-1}-k_2e_6)f,k_1+k_2-a(0),-k_1+6k_2-a(1),b-b(0)e_3,c-c(0)e_2f,d-d(0)f$,
  where $f$ is used to model $n!$, and eliminate the variables
  $k_1,k_2,e_{-1},e_2,e_3,e_6$ and $f$.
  \begin{alignat*}2
  &I(a,b,c,d) = \\
  &\;(I(2^n,3^n,1+6^n) + \langle a-(k_1e_{-1}-k_2e_6)f,k_1+k_2-a(0),\\
  &\quad -k_1+6k_2-a(1),b-b(0)e_3,c-c(0)e_2f,d-d(0)f\rangle)\\
  &\quad \cap\set K[a,b,c,d] =\\
  &\; \langle d(0)^2((-7b(0)c(0)a + a(0)bc)^2+a(1)bc(bc(a(1)+2a(0))-\\
&\quad  14b(0)c(0)a)) -(b(0)c(0)d(-6a(0)+ a(1)))^2\rangle.
  \end{alignat*}
  For instance, with the starting values $a(0)=2,a(1)=5$ and $b(0)=c(0)=d(0)=1$
  we get the relation
  \[b^2c^2 - 2abc + a^2 - d^2,\]
  with 
  \[a=((-1)^n+6^n)n!,\quad b=2^n,\quad c=3^nn!,\quad d=n!.\]
\end{example}

\begin{remark}
  Proposition~\ref{prop:mainprop} can easily be turned into an algorithm with
  the help of Gröbner bases, which allow computing a set of generators for the
  sum of ideals and also the elimination of variables. While computationally
  demanding, the use of Gröbner bases is viable in part because of the highly
  optimized tools that are available in modern computer algebra systems and in
  part because, as observed empirically in our experiments, the polynomial
  systems arising in practice in this context are typically small and easy to
  compute.
\end{remark}

\section{Implementation}
\label{sec:aligator}
The techniques presented in this paper are implemented in the open source
Mathematica software package Aligator\footnote{Aligator requires the Mathematica
packages Hyper~\cite{hyper}, Dependencies~\cite{dependencies} and
FastZeil~\cite{fastzeil}, where the latter two are part of the compilation
package ErgoSum~\cite{ergosum}.}~\cite{aligator}, available for download at
\begin{quote}\url{https://ahumenberger.github.io/aligator/}\end{quote}
We give an illustrative example of the provided facilities.
\begin{example} We compute the ideal of algebraic relations among the program
variables $a,b,c,d,e,f$ as given in the following loop. The loop exhibits two
first-order and two second-order recurrence relations ($a,e$ and $b,d$ resp.),
which Aligator could not handle before. Furthermore we have two first-order C-finite
recurrence relations ($c,f$). 

\vspace{5px}
\newenvironment{subq}{\begin{list}{}{%
\setlength\leftmargin{-21px}}%
\item}
{\end{list}}

\newenvironment{subq2}{\begin{list}{}{%
\setlength\leftmargin{-17px}}%
\item}
{\end{list}}

\begin{subq}{}
  \begin{mmaCell}[]{Input}
  Aligator[
    WHILE[True,
      a := 3(n + \mmaFrac{3}{2})a;
      s1 := s2; s2 := b;
      b := 5(\mmaFrac{3}{2} + n)s2 - \mmaFrac{3}{2}(1 + 2n)(3 + 2n)s1;
      c := -3c + 2;
      t1 := t2; t2 := d;
      d := 4(4 + n)t2 - 3(3 + n)(4 + n)t1;
      e := (n + 4)e; 
      f := 2f],
    LoopCounter -> n,
    IniVal -> \{
      t1 := 1; t2 := 1;
      s1 := 1; s2 := 2;
      a := 3; b := 1;
      c := 1; d := 3;
      e := 2; f := 5\}]
\end{mmaCell}
\end{subq}

The input is given to \mmaInlineCell{Input}{Aligator} in form of a while
loop, and two optional arguments: \mmaInlineCell{Input}{LoopCounter}
(default: \mmaInlineCell{Input}{i}) and \mmaInlineCell{Input}{IniVal} (default:
\mmaInlineCell{Input}{\{\}}). The former is for specifying which variable within
the loop corresponds to the loop counter, whereas the latter is for specifying
the initial values of the program variables. If no initial values are given,
then the invariants contain the starting values in the form of
\mmaInlineCell{Input}{a[0]}, representing the initial value of $a$.

The following output of \mmaInlineCell{Input}{Aligator} is a conjunction of the
elements of the Gröbner basis of the ideal of all algebraic relations among
$a,b,c,d,e$ and $f$. The loop counter is eliminated via Gröbner basis
computation.

\vspace{5px}
\begin{subq2}{}
\begin{mmaCell}{Output}
  2 d == 3 e && 
  (a + b) (2 d - 3 e) == 0 && 
  2 a d f == 3 a e f && 
  450 a b \mmaSup{(1 - 2 c)}{2} + 225 \mmaSup{b}{2} \mmaSup{(1 - 2 c)}{2} + 
    \mmaSup{a}{2} (225 \mmaSup{(1 - 2 c)}{2} - 16 \mmaSup{f}{2}) == 0
\end{mmaCell}
\end{subq2}
\end{example}

Note that the second and third invariant are consequences of the first one. By
setting the option \mmaInlineCell{Input}{GroebnerReduce -> True} a reduced
Gröbner basis is computed which does not contain redundant elements.

\section{Conclusion and Future Work}
\label{sec:conclusion}
We extended the class of P-solvable loops to include sums and products of
hypergeometric and C-finite sequences. This was made possible by identifying
algebraically independent factors in hypergeometric terms and then viewing the
sequences in question as rational function sequences over a transcendental field
extension. The implementation in Mathematica underlines the practicality of the
approach.

There are several promising directions in which we plan to expand this line of
research. Obviously, it is very desirable to include more types of recurrences
in P-solvable loops. These include further subclasses of the class of
holonomic sequences as well as partial and non-linear recurrence equations. It is
advisable to conduct a careful study on which kind of recurrences are relevant
in practice and also good-natured from a mathematical perspective. Uncoupling
techniques for systems of recurrence equations can also prove to be helpful in
this context.

Another possible extension is to consider nested loops. With the help of
$\Pi\Sigma^*$-theory~\cite{carsten2}, it might be possible to derive invariants
for the outermost loop, although the inner loops are not P-solvable by
themselves.

\bibliographystyle{ACM-Reference-Format} \bibliography{references}
\end{document}